\documentclass[11pt]{article}
\usepackage[english]{babel}
\usepackage[letterpaper,margin=1in]{geometry}
\usepackage{amsmath}
\usepackage{amsthm}
\usepackage{graphicx}
\usepackage{soul}
\usepackage{thmtools}
\usepackage{thm-restate}
\usepackage[colorlinks=true, allcolors=blue]{hyperref}
\usepackage{cleveref}
\usepackage{xcolor}
\definecolor{MyPurple}{RGB}{111,0,255}

\usepackage[utf8]{inputenc} 
\usepackage{url}            
\usepackage{booktabs}       
\usepackage{amsfonts}       
\usepackage{nicefrac}       
\usepackage{microtype}      
\usepackage{xcolor}

\usepackage{bbm}
\usepackage{amsthm}
\usepackage{amsmath}        
\usepackage{amssymb}        
\usepackage{comment}        
\usepackage{enumitem}  
\usepackage{parskip}
\usepackage{tikz}
\usepackage{subcaption}
\usetikzlibrary{calc}

\newtheorem{proposition}{Proposition}
\newtheorem{theorem}[proposition]{Theorem}
\newtheorem{definition}[proposition]{Definition}

\newtheorem{remark}[proposition]{Remark}
\newtheorem{lemma}[proposition]{Lemma}

\newtheorem{corollary}[proposition]{Corollary}

\newcommand{\OPT}{{\normalfont\text{OPT}}}

\newcommand{\bz}{{\mathbf z}}

\newcommand{\bg}{\boldsymbol g}

\newcommand{\bc}{\boldsymbol c}

\newcommand{\Acal}{\mathcal{A}}

\newcommand{\Ccal}{\mathcal{C}}

\newcommand{\Mcal}{\mathcal{M}}

\newcommand{\Ocal}{\mathcal{O}}

\newcommand{\Rcal}{\mathcal{R}}

\newcommand{\Zcal}{\mathcal{Z}}

\newcommand{\Nbb}{\mathbb{N}}

\newcommand{\Qbb}{\mathbb{Q}}
\newcommand{\Rbb}{\mathbb{R}}
\newcommand{\R}{\mathbb{R}}

\newcommand{\Zbb}{\mathbb{Z}}

\newcommand{\bx}{{\boldsymbol x}}

\newcommand{\by}{{\boldsymbol y}}

\newcommand{\oneb}{\boldsymbol{\mathbbm{1}}}

\newcommand{\1}{\mathbbm{1}}

\newcommand{\OT}{\text{OT}}
\newcommand{\Cost}{{\normalfont \texttt{Cost}}}

\newcommand{\ceil}[1]{{\lceil #1 \rceil}}
\newcommand{\floor}[1]{{\lfloor #1 \rfloor}}

\DeclareMathOperator*{\argmin}{arg\,min}

\usepackage[colorinlistoftodos,prependcaption,textsize=scriptsize]{todonotes}

\author{Christian Coester\\ University of Oxford \and Romain Cosson\\ New York University}
\title{Randomized $k$-server in polynomial time}

\begin{document}

\date{}
\maketitle

\begin{abstract}
We study the design of computationally efficient randomized algorithms for the $k$-server problem. Existing randomized algorithms with the best known competitive ratios are, on the one hand, inherently implicit and, on the other hand, employ a rounding scheme that maintains a distribution over exponentially many configurations. In this work, we introduce a derandomization framework that transforms any randomized $k$-server algorithm on a hierarchically separated tree into one that uses only $O(\log k)$ random bits for request sequences of arbitrary length—hence maintaining a distribution over only polynomially many server configurations. Leveraging this black-box derandomization, we obtain the first polynomial-time randomized $k$-server algorithm on arbitrary $n$-point metrics with a polylogarithmic competitive ratio. Our results also have implications for the advice complexity of the $k$-server problem.
\end{abstract}


\thispagestyle{empty}
\clearpage
\pagenumbering{arabic}

\section{Introduction}
The $k$-server problem is a cornerstone in online algorithms, sometimes described as the `holy grail of competitive analysis', and many techniques originally developed for $k$-server went on to impact various other problems in the field. While there has been limited progress on deterministic algorithms for the problem since a competitive ratio of $2k-1$ was established in \cite{koutsoupias1995k}, the study of randomized algorithms has seen recent breakthroughs, culminating in an $\Ocal(\log n\log^2k)$-competitive randomized algorithm on $n$-point metrics via an $\Ocal(\log^2k)$-competitive algorithm on hierarchically separated trees (HSTs) \cite{BubeckCLLM18}, and an $\Omega(\log^2 k)$ lower bound refuting the randomized $k$-server conjecture \cite{bubeck2023randomized}.

Despite the improved information-theoretic understanding of the problem, competitive algorithms that are also computationally efficient are lacking. Among algorithms for general metrics with polynomial (in the metric space encoding size) per-request running time, the best known competitive ratio is exponential in $k$, achieved by the randomized Harmonic algorithm~\cite{Grove91,BartalG00}. The work function algorithm~\cite{koutsoupias1995k}, which achieves the best deterministic competitive ratio, has time-per-request either exponential in $k$ or as an increasing function in the number of preceding requests (i.e., growing unbounded over time)~\cite{raghvendra2022scalable}.

For modern randomized algorithms, it is unclear how to even implement them at all while retaining competitiveness. Although the $k$-server problem is a discrete combinatorial online problem, recent progress relies on continuous methods such as existence results for solutions of differential inclusions, leading to implicit algorithms. The lack of efficiency of the algorithm in~\cite{BubeckCLLM18} was pointed out explicitly in Bubeck's TCS+ talk \cite[minute 40:50]{bubeck2017tcsplus}. This algorithm is defined by a continuous (both in time and space) projected gradient flow. A first step towards making the approach explicit---and towards obtaining a discrete algorithm that could be implemented---was taken by \cite{BuchbinderGMN19}, which replaced the continuous gradient flow with a sequence of Bregman projections, effectively discretizing the algorithm in time. In this paper, we address the problem of discretization in space. 
Specifically, we show that any randomized $k$-server algorithm on hierarchically separated trees (HSTs) can be transformed into an algorithm that uses a finite number of random bits (more precisely, $\Ocal(\log k)$ random bits, regardless of the length of the request sequence). While the basic idea of space discretizations has been present already in variants of paging~\cite{adamaszek2012log,bansal2021efficient}, prior techniques do not extend to $k$-server, as we discuss in Section~\ref{sec:failure}.

Our derandomization resolves two problems that hinder the implementability of current algorithms: first, that they require computing the exact solutions of a sequence of convex optimization problems, whereas finite time/space iterative methods can only provide approximate solutions; second, that they maintain a probability distribution over exponentially many server configurations (specifically, up to $\binom{n}{k}$ server configurations). In contrast, our approach reduces the cardinality of the support of a randomized algorithm to $\Ocal(k^2)$ -- a quantity that is independent of both the number of points in the metric space and the number of requests. 
This leads us to the first polynomial-time randomized $k$-server algorithm that achieves a polylogarithmic competitive ratio, matching the state-of-the-art competitiveness guarantee of \cite{BubeckCLLM18,BuchbinderGMN19}. Our results also have implications for the so-called `advice complexity' of the $k$-server problem \cite{bockenhauer2011advice}. 

\textit{$k$-server problem.} The problem is defined for a metric space $\Mcal = (M,d)$ with $n\geq k$ points. The input is a sequence of points $\rho(\cdot) = (\rho(1),\rho(2),\dots)$, where $\rho(t)\in M$ is called the \textit{request} at time $t\in \Nbb$. The output of the problem is a sequence of server configurations $\Ccal(\cdot) = (\Ccal(1),\Ccal(2),\dots)$, where $\Ccal(t)$ is a subset of $k$ points of $\Mcal$ that represent the position of $k$ servers at time $t$. The output must serve all the requests, meaning that it should satisfy $\forall t\colon \rho(t)\in \Ccal(t)$. 

\textit{Deterministic algorithms.} A deterministic (online) algorithm $\Ccal(\cdot)$ for the $k$-server problem is a function that maps a sequence of revealed requests $\rho(\leq t) =(\rho(1),\dots,\rho(t))$ to a configuration $\Ccal(\rho(\leq t)) \in M^k$ which contains the current request $\rho(t)$. In slight abuse of notation, we denote the configuration $\Ccal(\rho(\leq t))$ by $\Ccal(t)$. The cost of a deterministic algorithm $\Ccal(\cdot)$ on the input sequence $\rho(\cdot)$ is defined by 
\begin{equation*}
    \Cost(\Ccal(\cdot),\rho(\cdot)) = \sum_{t} d(\Ccal(t),\Ccal(t+1)),
\end{equation*}
where the distance between two configurations $d(\Ccal,\Ccal')$ is defined as the minimum weight of a matching between $\Ccal$ and $\Ccal'$. 

\textit{Randomized algorithms.} A randomized (online) algorithm $\Rcal(\cdot)$ is a probability distribution over deterministic (online) algorithms. The cost of a randomized algorithm on an instance is the expected cost on that instance. For a parameter $m\in \Nbb$, an \emph{$m$-barely random algorithm} is defined as the uniform distribution over $m$ deterministic algorithms $\Rcal(\cdot) = (\Ccal_1(\cdot),\dots,\Ccal_m(\cdot))$, and its cost on an instance $\rho(\cdot)$ is therefore given by $\Cost(\Rcal(\cdot),\rho(\cdot)) = \frac{1}{m}\sum_{j\in [m]}\Cost(\Ccal_j(\cdot),\rho(\cdot))$. Note that if $m=2^b$ for an integer $b$, an $m$-barely random algorithm can be sampled with just $b$ random bits, independently of the length of the input sequence--hence the term \textit{barely} random algorithm, introduced in \cite{reingold1994randomized}.

\textit{Advised algorithm.} An online algorithm with advice $\Acal(\cdot)$ is a deterministic algorithm whose input is augmented with some bits of advice supposedly provided by a third party that knows the problem instance. 
For any parameter $m\in\Nbb$, we define an $m$-advised algorithm $\Acal(\cdot)$ as a collection of $m$ deterministic algorithms $\Acal(\cdot)=(\Ccal_1(\cdot),\dots,\Ccal_m(\cdot))$. The cost of an $m$-advised algorithm is defined as 
$\Cost(\Acal(\cdot),\rho(\cdot)) = \min_{j\in [m]} \Cost(\Ccal_j(\cdot),\rho(\cdot))$, and is therefore not greater than its barely random counterpart.

\textit{Competitive analysis.} The performance of an online algorithm for the $k$-server problem is evaluated by its competitive ratio. We say that an online algorithm $\Acal(\cdot)$ is $c$-competitive if it satisfies for any input $\rho(\cdot)$ that 
\begin{equation}\label{eq:competitiveness}
\Cost(\Acal(\cdot),\rho(\cdot))\leq c \cdot \text{OPT}(\rho(\cdot))+c_0,
\end{equation}
where $c$ is a positive real number, $\text{OPT}(\rho(\cdot)) = \inf_{\Ccal(\cdot)} \Cost(\Ccal(\cdot),\rho(\cdot))$ is called the offline optimum and $c_0$ is a constant independent of the sequence of requests.

\paragraph{Main result.} The main result of this paper (Theorem \ref{th:main-complexity}) is a polynomial-time randomized $k$-server algorithm that enjoys state-of-the-art competitive guarantees. The main technical ingredient for this result is a polynomial-time derandomization scheme (Proposition \ref{prop:frac-barely-frac} and Proposition \ref{prop:barely-frac-rand}) transforming a $c$-competitive randomized $k$-server algorithm on hierarchically separated trees (HSTs) into an $\Ocal(c)$-competitive $k$-server algorithm that uses only $\Ocal(\log k)$ random bits. As detailed in the related works section below, this framework extends previous results on the randomized complexity \cite{bockenhauer2009advice} and the advice complexity \cite{emek2011online,bockenhauer2011advice} of the $k$-server problem.

\subsection{Previous work} \cite{manasse1988competitive} introduced the $k$-server problem as a general framework for online problems. 
This initial work focused on the competitive ratio of deterministic algorithms, establishing a universal lower bound of $k$ and conjecturing that a $k$-competitive algorithm exists for all metric spaces. The problem quickly became central in the domain of online algorithms, and, years after, the deterministic analysis of the $k$-server problem was nearly settled by \cite{koutsoupias1995k}, which proved that the work function algorithm is $(2k-1)$-competitive in all metric spaces. In this paper, we are primarily concerned with the randomized and the fractional variants of the $k$-server problem, and we refer to \cite{koutsoupias2009k} for details on the deterministic variant. 

The first randomized analysis of the $k$-server problem was by \cite{fiat1991competitive}. They introduced a simple (and discrete) `marking' algorithm that is $2H_k$-competitive on the uniform metric space (equivalent to the paging problem). In contrast, going beyond the uniform metric space later required the introduction of the fractional variant of the $k$-server problem. The idea is to have the algorithm maintain a (fractional) measure over the metric space, of total mass $k$, representing for each point the probability that a server is located at that point. This effectively allows the algorithm to move probability mass continuously in space, enabling the use of new methods. Any randomized algorithm straightforwardly induces a fractional algorithm, but a reciprocal transformation is only known for the line metric \cite{dehghani2017stochastic} and hierarchically separated trees (HSTs) \cite{bansal2015polylogarithmic} (as well as metrics that embed into HSTs with constant distortion, such as weighted stars, for which this result was established first \cite{bansal2012primal}). This typically requires a nontrivial `rounding' step. 
The idea of studying the fractional variant of the $k$-server problem was first used by \cite{blum1999finely} to obtain refined randomized guarantees on the uniform metric space. 
In subsequent breakthroughs, first \cite{bansal2012primal} used the approach to obtain an $\Ocal(\log k)$-competitive algorithm on weighted stars with a primal-dual algorithm. After, \cite{bansal2015polylogarithmic} obtained an $\tilde\Ocal(\log^3n\log^2k)$-competitive algorithm $n$-point metrics, based on a fractional algorithm on HSTs. To date, the best randomized competitive ratio is due to \cite{BubeckCLLM18}, which presented an $\Ocal(\log^2 k)$-competitive algorithm for HSTs. The focus on hierarchically separated trees is theoretically motivated by random tree embeddings \cite{fakcharoenphol2003tight}, which allow to embed any $n$-point metric space into a hierarchically separated tree with a distortion of $\Ocal(\log n)$. Therefore, the result of \cite{BubeckCLLM18} directly translates into an $\Ocal(\log n \log^2 k)$-competitive algorithm on arbitrary metrics. The question whether an $o(k)$-competitive randomized algorithm exists for the $k$-server problem on arbitrary metrics remains open.

The study of randomized complexity in online algorithms started with the work of \cite{reingold1994randomized} on the list update problem. They noticed that a finite number of random bits (independent of the length of the request sequence) is sufficient to improve the competitive ratio from $2$ to $1.75$ and coined the term `barely random algorithm'. The notion is also covered in \cite[Chap. 2]{borodin2005online}. In the case of the $k$-server problem, a result of this type is in \cite{bockenhauer2009advice,komm2011advice} (see also \cite[Chap. 3]{komm2016introduction}) and is restricted to the uniform metric space. Their result essentially says that there exists an $\Ocal(\log k)$-competitive algorithm for paging that requires only $\Ocal(\log k)$ random bits. This can be recovered as a consequence of our derandomization method since there exist $\Ocal(\log k)$-competitive randomized algorithms on the uniform metric space \cite{fiat1991competitive}. Other works, such as \cite{emek2011online,bockenhauer2011advice, renault2015online}, study the randomized complexity or advice complexity of the $k$-server problem, but in a weaker regime where the number of bits of advice or the number of random bits is allowed to scale with the size of the problem instance. For example, \cite{emek2011online} introduces a deterministic $\Ocal(\sqrt{k})$-competitive algorithm for the $k$-server problem that uses $\Ocal(1)$ bits of advice \emph{per request}.

One important technical tool that we use throughout the paper is the notion of $m$-barely fractional algorithms, which allows us to recover some of the power of fractional algorithms, while enforcing a discretization in space. Informally, an $m$-barely fractional algorithm is one that moves server mass in fractions that are multiples of $1/m$, and consequently, the size of its support is bounded by $km$. Furthermore, applying the rounding of \cite{bansal2015polylogarithmic}, the cardinality of the support of the distribution over deterministic algorithms of the corresponding randomized algorithm can be restricted to $m$. We note that barely fractional algorithms are already instrumental in previous works, such as in \cite{adamaszek2012log,bansal2021efficient}, which study variants of the $k$-server problem on weighted stars (cf. Section~\ref{sec:failure} for further discussion), and in \cite{cosson2024barely}, which studies a multi-agent version of metrical task systems.

\subsection{Preliminaries}

For $x\in \R$, we write $x^+:=\max\{x,0\}$.

\paragraph{Weighted trees.} Let $V$ be the set of nodes of a rooted edge-weighted tree. We denote by $r\in V$ the root and by $L\subseteq V$ the set of leaves. For a node $u\in V\setminus\{r\}$, we denote by $p(u)$ its parent and by $w_u$ the weight of the edge between $u$ and $p(u)$.
For a node $u\in V$, we denote by $C_u$ the set of its children and by $L_u$ the set of leaves in the subtree rooted at $u$.

Any weighted tree induces a metric space on its leaves $L$,\footnote{We do not consider internal nodes to be points of the metric space, and in particular we assume that $k$-server requests never come at internal nodes. This is without loss of generality, as new leaves could be attached to internal nodes with edges of length $0$.} where the distance between two leaves is defined as the sum of edge weights on their connecting path.

For $\tau\ge 1$, a $\tau$-HST (hierarchically separated tree) is the metric space induced by a tree whose edge weights satisfy $w_u =\tau w_{v}$ for all $u \in V\setminus\{r\}$ and $v\in C_u$, and where all leaves are at the same depth. It is well-known that any $n$-point metric space can be probabilistically embedded into an HST such that all distances are distorted by a factor $O(\log n)$ in expectation~\cite{fakcharoenphol2003tight}.

\paragraph{Inner and leaf measures.} An \emph{inner measure} of total mass $k$ on a tree with nodes $V$ is a vector $\bz\in \Zcal$, where 
$$\Zcal = \left\{\bz \in \Rbb_+^V \text{ s.t. } z_r = k \text{ and } \forall u\in V\colon z_u \geq \sum_{c\in C_u} z_c \right\}.$$
Here, $z_u$ represents the amount of server mass in the subtree rooted at $u$. We write $z_{[u]} = z_u-\sum_{c\in C_u}z_c$ for the mass at $u$ itself. If $z_{[u]}=0$ for every $u\in V\setminus L$, then $\bz$ is called a \emph{leaf measure}. Thus, a leaf measure is a measure of mass $k$ on the points $L$ of the associated metric space, whereas an inner measure may have some of the mass on internal nodes, which we do not consider part of the metric space. We say that a measure is \emph{$m$-barely fractional} if it only takes values in $\frac{1}{m}\Nbb = \{0,\frac{1}{m},\frac{2}{m},\dots\}$.

\paragraph{Optimal transport.} We define the transport distance between two measures $\bz,\bz'\in\Zcal$ by 
$$\OT(\bz,\bz') = \sum_{u\in V\setminus\{r\}}w_u|z_u-z_u'|. $$
This definition matches the classical interpretation of optimal transport for measures on trees. 

\paragraph{Fractional algorithm.} A fractional (online) $k$-server algorithm on a tree metric $L$ is a function $\bz(\cdot)$ that takes as input a sequence $\rho(\leq t)$ and produces as output a leaf measure denoted by $\bz(t)$, with total mass $\sum_{u\in L} z_u(t) = k$ and that services requests by satisfying $z_{\rho(t)}(t) = 1$. The cost of a fractional algorithm $\bz(\cdot)$ on some instance $\rho(\cdot)$ is defined by 
$$\Cost(\bz(\cdot),\rho(\cdot)) = \sum_{t}\OT(\bz(t+1),\bz(t)).$$ A fractional algorithm $\bz(\cdot)$ is $m$-barely fractional if it only outputs $m$-barely fractional measures.

\paragraph{Computational complexity of algorithms.} To properly define computational complexity, we restrict our attention to metric spaces where all distances are integral. This assumption is without loss of generality, up to arbitrary precision, because the $k$-server problem is scale-invariant. We use $D$ to denote the diameter of the underlying metric space. Thus, distances are encoded in words of $O(\log D)$ bits. We aim for algorithms that process each new request in time polynomial in $\log D$ and the number of points $n$ in the metric space.

\subsection{\texorpdfstring{Why past approaches do not extend to $\boldsymbol{k}$-server}{Why past approaches do not extend to k-server}}\label{sec:failure}

The idea of working with barely fractional algorithms is already present in \cite{adamaszek2012log}, which applies it to weighted/generalized paging, corresponding to star metrics for $k$-servers. In this case, a barely fractional algorithm can be achieved very simply by rounding each (fully fractional) evicted page mass~$x$ to $f(x)=\frac{1}{k}\lfloor \min(2kx,k)\rfloor$. However, discretizations of this form are not competitive for $k$-server, because tiny fluctuations in $x$ may lead to large movements in $f(x)$ (because $\floor{\cdot}$ is discontinuous!). This issue does not arise on stars, where the evicted page mass $x$ is decreased only when set to $0$, but on trees, it occurs at internal nodes. In fact, all approaches of the form $x\mapsto f(x)$ for a fixed function $f(\cdot)$ are condemned to fail.

Another possible perspective on the rounding of \cite{adamaszek2012log} for star metrics is to assume that the root can hold server mass and to consider the modified algorithm in which upward movements (from a leaf to the root) are twice those of the base algorithm\footnote{Movements in the modified algorithm stop when there is no mass available at the source.} and where requests are served by pulling server mass down from the root to the requested leaf. One can show that the modified algorithm is valid and suffers only a factor-of-2 loss in the movement cost. Furthermore, this transformation naturally extends to metric spaces that can be embedded as the leaves of a tree, by assuming that all non-leaf nodes can hold server mass and that requests are served by pulling server mass down from the nearest ancestors of the requested leaf. In star metrics, the cumulative movements of the modified algorithm can be rounded down to the nearest multiple of $1/k$, yielding exactly the rounding of~\cite{adamaszek2012log}. But in trees, that simple rounding is no longer possible: Consider a node $u$ with two children, each of which sends mass $\frac{1}{2k}$ to its grandparent in the accelerated algorithm (i.e., to $p(u)$). Then the mass traveling from $u$ to $p(u)$ is $\frac{1}{k}$, so we would want to send the same mass after discretization. But each movement from a child of $u$ to $u$ is only $\frac{1}{2k}$, which might be insufficient to cross the next multiple of $\frac{1}{k}$. Hence, the discretized algorithm would not have any movement into $u$ that could be sent further upwards to $p(u)$.

We resolve these kinds of issues via a technique inspired by physics that we can refer to as \textit{hysteresis} (see Section \ref{sec:barely-frac}). Roughly, it ensures that the barely fractional configuration only changes when the proximity to the target configuration improves a lot, i.e., when moving is really worth it.

The question of computational complexity of variants of weighted paging was also raised in a more recent paper of Bansal et al. \cite{bansal2021efficient}, which proposed a polynomial-time algorithm for a generalization of weighted paging at the expense of an overall competitiveness of $O(\log^2 k)$ instead of $O(\log k)$.

\subsection{Exposition of the results and outline}
Our main result is the following.
\begin{restatable}{theorem}{mainComplexity}\label{th:main-complexity}
There exists an $\Ocal(\log^2 k)$-competitive randomized $k$-server algorithm for $\tau$-HST metrics with $\tau\geq 10$ that serves requests in polynomial time in $(n,\log D)$. The algorithm uses only $O(\log k)$ random bits, regardless of the length of the request sequence.
\end{restatable}
By classical polynomial time random HST embeddings~\cite{fakcharoenphol2003tight}, we obtain the following extension of the theorem to general metrics.
\begin{corollary}\label{cor:generalMetrics}
There exists a randomized $\Ocal(\log n \log^2 k)$-competitive $k$-server algorithm on arbitrary $n$-point metrics that serves requests in polynomial time. 
\end{corollary}
We note that this algorithm for general metrics is still barely random as it only samples a bounded number of random bits at the start of the program. Most of these bits are used to compute the HST embedding.

By observing that bits of advice are at least as useful as random bits (see discussion above), we also obtain the following corollary. 
\begin{corollary} For any $\tau$-HST metric, with $\tau\geq 10$, there exists a deterministic $k$-server algorithm with advice that is $\Ocal(\log^2 k)$-competitive and that only uses $\Ocal(\log k)$ bits of advice. 
\end{corollary}
We now streamline the main arguments leading to Theorem \ref{th:main-complexity}. We start with the fractional algorithm of \cite{BuchbinderGMN19}, that can be adapted to run in polynomial time albeit at the expense of an additive penalty, as stated below. 
\begin{restatable}{proposition}{approximateSmile}
\label{prop:approximate-smile}
Let $\epsilon>0$. There exists a 
fractional $k$-server algorithm $\bz(\cdot)$ for $\tau$-HST metrics with $\tau\geq 10$ that serves each request in polynomial time in $(n,\log 1/\epsilon, \log D)$, and satisfying for all requests sequences $\rho(\cdot)$
$$\Cost(\bz(\cdot),\rho(\cdot))\leq \Ocal(\log^2 k) \OPT(\rho(\cdot))  + \epsilon T +O(D\log^2 k),$$
where $T$ is the length of the sequence $\rho(\cdot)$.
\end{restatable}
Note that this fractional algorithm is not competitive in the sense of Equation \eqref{eq:competitiveness}. The seemingly innocuous additive term in $\epsilon T$ is, in fact, quite inconvenient because a malicious adversary can make the additive term $\epsilon T$ very large, while keeping the offline cost bounded. Specifically, consider the case of a $(k+1)$-point metric with $k$ points close to one another, and one far-away point, and assume only the $k$ nearby points are requested. A natural fractional algorithm (and indeed the algorithm in \cite{BuchbinderGMN19}) would send only a very small (smaller than $\epsilon$) fractional mass from the far-away point to the group of $k$ points, per request. Therefore, a polynomial-time implementation that approximates fractional mass to within precision $\epsilon$ in each step would never move any mass from the far-away point to the group of $k$ points, and thus, would not have a finite competitive ratio. The other crucial issue with this fractional algorithm is that, under the rounding of \cite{bansal2015polylogarithmic}, it induces a randomized algorithm with exponential support, thereby losing its computational efficiency. We deal with both issues by discretizing the fractional algorithm (Proposition \ref{prop:frac-barely-frac}) before applying the rounding of \cite{bansal2015polylogarithmic} (Proposition \ref{prop:barely-frac-rand}). For further generality, the statements of the propositions are extended in Proposition~\ref{prop:frac-barely-frac} to arbitrary trees, and, in Proposition~\ref {prop:barely-frac-rand}, to all metric spaces for which a rounding scheme is known. The combination of Proposition~\ref{prop:frac-barely-frac} and Proposition \ref{prop:barely-frac-rand} entails our derandomization scheme.

\begin{restatable}[Fractional to barely fractional]{proposition}{barelyfractional}
\label{prop:frac-barely-frac}
Let $m\geq 2k^2+k$. Any fractional online algorithm for $k$-server on trees can be converted online into an $m$-barely fractional algorithm with at most a polynomial additive overhead in running time and a constant factor multiplicative overhead in cost.
\end{restatable}

\begin{restatable}[Barely fractional to barely randomized]{proposition}{barelyrandom}
\label{prop:barely-frac-rand}
    For $k$-server on a line metric 
    or $\tau$-HSTs with $\tau> 5$, any $m$-barely fractional algorithm can be converted online into an $m$-barely random algorithm with at most a polynomial additive overhead in running time and a constant factor multiplicative overhead in cost.
\end{restatable}
It is clear from the statement of these propositions why they lead to a polynomial time randomized algorithm (provided that $1/\epsilon$ and $m$ are sufficiently `small'), but it is not immediately obvious why they also solve the issue that was initially raised by the additive term. For this, we note that we can pick $\epsilon = \Theta\left(\frac{1}{m}\right)$, and then directly apply the following result. 

\begin{restatable}{proposition}{propAdditive}\label{prop:additive}
    Suppose that there is an $m$-barely fractional algorithm $\bz(\cdot)$ satisfying for any request sequence $\rho(\cdot)$,
    $$\Cost(\bz(\cdot), \rho(\cdot))\leq c\, \OPT(\rho(\cdot)) + \frac{1}{2m} T + c_0,$$
    where $T$ is the number of time-steps in the input sequence $\rho(\cdot)$, $c_0$ is a constant independent of $\rho(\cdot)$, and each non-zero distance is at least $1$. Then the variant of $\bz(\cdot)$ that ignores any request that it can serve without moving is $2c$-competitive. 
\end{restatable}

We prove Propositions~\ref{prop:frac-barely-frac}~and~\ref{prop:additive} in Section \ref{sec:barely-frac}, Proposition \ref{prop:barely-frac-rand} in Section \ref{sec:barely-rand}, and Proposition~\ref{prop:approximate-smile} in Section~\ref{sec:ccomplexity}.

\section{\texorpdfstring{Barely fractional $k$-server}{Barely fractional k-server}}\label{sec:barely-frac}
In this section, we prove the main technical step of this paper of going from fractional to barely-fractional algorithms. Before presenting the proof, we start with a general discussion of how \textit{hysteresis} can be used in competitive analysis, motivating our approach with a simple example inspired by dynamic data structures.

\barelyfractional*

\begin{remark}[Collective interpretation of $m$-barely fractional algorithms] One way to interpret $m$-barely fractional $k$-server is as an online problem with $mk$ deterministic agents, where each incoming request must be served by moving $m$ agents to the requested point. 
This perspective connects to the metric allocation problem \cite{BansalC22}, which generalizes $k$-server by allowing a broader class of requests. It also relates to collective metrical task systems \cite{cosson2024barely}, where a team of $m$ agents jointly faces a metrical task system, and to collective tree exploration \cite{cosson2024collective}, where a team of $k$ agents must explore an unknown tree as quickly as possible.
\end{remark}

\paragraph{Warm-up: competitiveness via hysteresis.} Hysteresis is a phenomenon where a system's present state depends on its history. It is a common theme in physics, specifically in electromagnetism, where it describes the fact a magnet can resist a (mildly) opposing magnetic field. 

In computer science, the concept appears in dynamic data structures \cite{bender2023iceberg}, and in real-world implementations of resizable arrays (see e.g., \cite{berger1999hoard, schildermans2020ptlbmalloc2}). Specifically, consider the situation where a user wishes to store at time $t$ an array of size $y(t)\in \Nbb$. Since reallocating memory space is costly, a common technique is to allocate at any time $t$ an array of size $x(t) \geq y(t)$, where the value of $x(t)$ is more `stable' than the value of $y(t)$ (i.e., it resists change). A typical update rule for $x(t)$ is as follows:
\begin{equation}
    \begin{cases}
        \text{If } y(t)\ge x(t-1) \text{ then set } x(t) = 2x(t-1).\\
        \text{If } y(t)\le x(t-1)/4 \text{ then set } x(t) = x(t-1)/2.
    \end{cases}\label{eq:basic-trick}
\end{equation}
Basic arguments involving geometric sums demonstrate the approach's competitiveness. 
As a warm-up to our proof of Proposition \ref{prop:frac-barely-frac}, we propose a new perspective on the update rule described in Equation \eqref{eq:basic-trick}, which can equivalently be formulated as 
\begin{equation}
    x(t) = \argmin_{x\in \{2^i\colon i\in \Nbb\}} |x-x(t-1)|+|x-2 y(t)|,\label{eq:basic-hysteresis}
\end{equation}
where ties are broken in favor of the solution maximizing $|x-x(t-1)|$. The idea is that the first term $|x-x(t-1)|$ resists changes in the value of $x$, whereas the other term $|x-2 y(t)|$ encourages $x$ to take a value close to $2y$. We now show the competitiveness of Equation \eqref{eq:basic-hysteresis} without using geometric sums.  By optimality of $x(t)$ in \eqref{eq:basic-hysteresis} we have that
\begin{equation}\label{eq:basic-hysteresis-argument}
    |x(t)-x(t-1)| + |x(t)-2y(t)|= |x(t-1)-2y(t)|. 
\end{equation}
Subtracting $|x(t-1)-2y(t-1)|$ from both sides of Equation \eqref{eq:basic-hysteresis-argument}, and using on the right-hand side that $|x(t-1)-2y(t)|-|x(t-1)-2y(t-1)|\leq 2|y(t-1)-y(t)|$, and on the left-hand side, the telescopic sum $\sum_{t\in [T]}|x(t)-2y(t)|-|x(t-1)-2y(t-1)| = |x(T)-2y(T)|-|x(0)-2y(0)|\geq 0$, we get,
\begin{equation*}
    \sum_{t\in [T]}|x(t)-x(t-1)|  \leq 2\sum_{t\in [T]}|y(t)-y(t-1)|,
\end{equation*}
and thus, the competitiveness of $\bx(\cdot)$.
This analysis can be extrapolated to more complex settings, as we will see in the proof below. Specifically, the reader will notice the parallel between Equation~\eqref{eq:basic-hysteresis} and our discretization technique for $k$-server in Equation \eqref{eq:z2def}, where the absolute value is replaced with an optimal transport cost, and where the set $\{2^i\colon i\in \Nbb\}$ is replaced by a set of barely fractional configurations.

\begin{proof}[Proof of Proposition \ref{prop:frac-barely-frac}] Let $\bx(\cdot)$ be the given fractional algorithm, and assume without loss of generality that it is initialized with an integral configuration $\bx(0)$. The proof goes through a sequence of steps that progressively go from the initial fully fractional strategy $\bx(\cdot)$ to the desired barely fractional strategy $\by(\cdot)$:
\begin{enumerate}
    \item First, we define a sequence of inner measures $\bz^{(1)}(\cdot)$ of total mass $k$, with a movement cost at most twice the movement cost of $\bx(\cdot)$, and satisfying at all times and for all nodes $u\in V\colon z_{u}^{(1)}>0 \implies x_u\geq 1/2$ as well as $x_u\geq 1\implies z_u^{(1)}\geq 1$, where we drop the dependency on $t$ for convenience. 
    \item Next, for $m' = 2m+2k+1$, we define a sequence of $m'$-barely fractional inner measures $\bz^{(2)}(\cdot)$ of total mass $k$, of movement cost at most the movement cost of $\bz^{(1)}(\cdot)$, and satisfying at all times and for all nodes $u\in V\colon z^{(2)}_u \geq z^{(1)}_u -(2k+1)/m'$.
    \item Then, we define a sequence of $2m$-barely fractional inner measures $\bz^{(3)}(\cdot)$ of total mass in the interval $[k,k+1/2]$ (the same mass in each step), of movement cost at most twice the movement cost of $\bz^{(2)}(\cdot)$ and satisfying at all times, and for all nodes $u\in V\colon x_u\geq 1 \implies z_u^{(3)}\geq 1$.
    \item After, we define a sequence of $m$-barely fractional inner measures $\bz^{(4)}(\cdot)$ of total mass $k$, of movement cost at most twice the movement cost of $\bz^{(3)}(\cdot)$, and satisfying at all times and for all nodes $u\in V\colon x_u\geq 1 \implies z_u^{(4)}\geq 1$.\label{it:kPlusHalfToK}
    \item Finally, observing that all prior steps can be performed online, we obtain our desired $m$-barely fractional online algorithm by converting the sequence of inner measures $\bz^{(4)}(\cdot)$ into a sequence of leaf measures. 
\end{enumerate}

\paragraph{Step 1.} We consider the function $\sigma \colon \Rbb_+ \rightarrow \Rbb_+$ such that $\sigma\vert_{\left[\ell,\ell+\frac{1}{2}\right]} = \ell$ for every $\ell\in \Zbb_+$ and $\sigma$ is extended affinely to the rest of $\Rbb_+$. Equivalently, $\forall x\in \Rbb_+\colon \sigma(x) = \floor{x}+2(x-\floor{x}-1/2)^+$. The definition is borrowed from \cite{BubeckCLLM18}. We consider the strategy $\bz^{(1)}(\cdot)$ defined at all times by $\bz^{(1)}(t) = \sigma(\bx(t))$ where $\sigma(\cdot)$ is applied element-wise (including to inner nodes of the tree). Note that $\bz^{(1)}(\cdot)$ always defines an inner measure of total mass $k$ because $\sigma(\cdot)$ is superadditive and $\forall t\colon z^{(1)}_r(t) = \sigma(k) = k$. Also, since $\sigma(\cdot)$ is $2$-Lipschitz, the movement cost of $\bz^{(1)}(\cdot)$ is at most twice the cost of $\bx(\cdot)$.  

\paragraph{Step 2.} Next, we consider the strategy $\bz^{(2)}(\cdot)$ defined by initializing $\bz^{(2)}(0)=\bz^{(1)}(0) =\bx(0)$ and with the induction for $t\geq 1$, 
\begin{equation}\label{eq:z2def}
    \bz^{(2)}(t) \in \argmin_{\substack{\bz \in \Zcal \\ \forall u\in V\colon z_u \in \frac{1}{m'}\Nbb}} \OT(\bz^{(2)}(t-1),\bz)+\OT(\bz,\bz^{(1)}(t)),
\end{equation}
where we split ties in favor of a configuration maximizing $\OT(\bz^{(2)}(t-1),\bz)$. Note that the old configuration $\bz=\bz^{(2)}(t-1)$ already minimizes \eqref{eq:z2def}, due to the triangle inequality. Starting from this $\bz$, in a polynomial number of steps we can successively modify $\bz$ until it satisfies the tie-breaking rule, as follows: While the tree contains an edge $(u,v)$ such that $z_{[u]}\ge\frac{1}{m'}$ and the subtree $V_u$ of nodes closer to $u$ than $v$ satisfies $\sum_{s\in V_u}z_{[s]}\ge \sum_{s\in V_u}z^{(1)}_{[s]} + \frac{1}{m'}$, we change $\bz$ by removing mass $\frac{1}{m'}$ from $u$ and adding it to $v$. This increases the first term in the optimization problem~\eqref{eq:z2def} by $\frac{1}{m'}$ times the length of the edge $(u,v)$ and decreases the second term by the same amount. Once such an edge $(u,v)$ no longer exists, we set $\bz^{(2)}(t):=\bz$. In other words, we gradually move $\bz$ away from $\bz^{(2)}(t-1)$ and towards $\bz^{(1)}(t)$ as long as another $m'$-barely fractional configuration in that direction exists. Polynomial running time of the procedure follows from the fact that each edge $(u,v)$ can be chosen at most $m'k$ times.

We first show that the movement cost of $\bz^{(2)}(\cdot)$ is at most the movement cost of $\bz^{(1)}(\cdot)$. We have by optimality of $\bz^{(2)}(t)$ in \eqref{eq:z2def} and triangle inequality that
\begin{equation*}
    \OT(\bz^{(2)}(t-1),\bz^{(2)}(t))+\OT(\bz^{(2)}(t),\bz^{(1)}(t))) = \OT(\bz^{(2)}(t-1),\bz^{(1)}(t)).
\end{equation*}
By the triangle inequality $\OT(\bz^{(2)}(t-1),\bz^{(1)}(t)) \leq \OT(\bz^{(2)}(t-1),\bz^{(1)}(t-1))+\OT(\bz^{(1)}(t-1),\bz^{(1)}(t))$, we then obtain the inequality
\begin{equation*}
    \OT(\bz^{(2)}(t-1),\bz^{(2)}(t))+\OT(\bz^{(2)}(t),\bz^{(1)}(t))-\OT(\bz^{(2)}(t-1),\bz^{(1)}(t-1)) \leq \OT(\bz^{(1)}(t-1),\bz^{(1)}(t)).
\end{equation*}
We therefore conclude that the movement cost of $\bz^{(2)}(\cdot)$ is less than the movement cost of $\bz^{(1)}(\cdot)$ by taking a telescopic sum in the above equation. 

Second, we prove that for all nodes $u\in V\colon z^{(2)}_u(t) \geq z^{(1)}_u(t)-(2k+1)/m'$. For this, we denote by $F_u$ the set of nodes $v\in V\setminus\{u\}$ satisfying $z^{(2)}_{[v]}>0$ and such that $\bz^{(2)}$ has no mass on the path from $v$ to $u$ (except at its endpoints). Denoting by $\pi(\cdot,\cdot)$ an optimal transport from $\bz^{(1)}$ to $\bz^{(2)}$, we have that 
\begin{align}
    z^{(2)}_u &= z^{(1)}_u + \sum_{v\in V}\pi(v,u) - \sum_{v\in V}\pi(u,v),\nonumber\\
    & \geq z^{(1)}_u- \sum_{v\in V\setminus\{u\}}\pi(u,v),\nonumber\\
    &= z^{(1)}_u- \sum_{v\in F_u}\left(\sum_{w\in V(v)} \pi(u,w)\right),\label{eq:F}\\
    &\geq z^{(1)}_u -\frac{|F_u|}{m'}.\label{eq:no-move}
\end{align}
where in \eqref{eq:F}, we denote by $V(v)$ the tree composed of nodes that are separated from $u$ by $v\in F_u$, and where for \eqref{eq:no-move} we use that $\forall v\in F_u\colon \sum_{w\in V(v)} \pi(u,w) \leq 1/m'$, otherwise, the tie-breaking rule of \eqref{eq:z2def} is contradicted by considering the configuration $\bz$ where, starting from $\bz^{(2)}$, we remove a mass of $1/m'$ from $v$ and add it to $u$ instead. Then, we observe that (a) at most one node of $F_u$ is an ancestor of $u$, and we denote it by $\bar u$, and (b) the set $F_u' = F_u\setminus \{\bar u \}$ does not contain a pair of nodes where one is an ancestor of the other. Since $\forall v\in F_u' \colon z^{(2)}_v >0$ we have that $z^{(1)}_v>0$ and thus that $x_v \geq 1/2$. Thus, all elements of $F_u'$ correspond to a mass of at least $1/2$ for $\bx$, which is of total mass $k$. This shows that $|F_u'|\leq 2k$, and thus that $|F_u|\leq 2k +1$.

\paragraph{Step 3.} We then consider the sequence defined at all times by $\bz^{(3)}(t) = \lambda \bz^{(2)}(t)$, for $\lambda = \frac{1}{1-(2k+1)/m'}$,   which defines $2m$-barely fractional measures of mass $\lambda k$. We note that the constraint $4k^2+4k+1 \leq m'$, i.e.,  $2k^2+k\leq m$, effectively enforces that $\lambda k\in [k,k+1/2]$. Clearly, since $\lambda\leq 2$, the movement cost of $\bz^{(3)}(\cdot)$ is at most twice the movement cost of $\bz^{(2)}(\cdot)$. Also, for any node $u\in V$ for which $x_u\geq 1$, we have that $z_u^{(2)}\geq z_u^{(1)}-\frac{2k+1}{m'}\geq 1-\frac{2k+1}{m'} = 1/\lambda$ and thus $z_u^{(3)} = \lambda z_u^{(2)} \geq 1$. 

\paragraph{Step 4.} We finally consider $\bz^{(4)}(\cdot)$ defined by $\forall t\colon \bz^{(4)}(t) = \sigma(\bz^{(3)}(t))$. This defines a valid sequence of inner measures because $\sigma(\cdot)$ is super-additive, with mass $k$ since we have that $z^{(4)}_r = \sigma\left(\lambda k\right) = k$. Also, we clearly have that whenever $x_u\geq 1$, $z^{(3)}_u\geq 1$ and thus, $z^{(4)}_u\geq 1$. Each measure $\bz^{(4)}(t)$ is $m$-barely fractional because $\bz^{(3)}(t)$ is $2m$-barely fractional and $\sigma\left(\frac{1}{2m}\Nbb\right)\subset \frac{1}{m}\Nbb$. By the same Lipschitzness argument as above, the movement cost of $\bz^{(4)}(\cdot)$ is at most twice the movement cost of $\bz^{(3)}(\cdot)$.

\paragraph{Step 5.} We observe that all the definitions above can be implemented online and in polynomial time. The algorithm $\bz^{(4)}(\cdot)$ satisfies almost all the desired properties: It is $m$-barely fractional, has a total mass $k$, serves the same requests as $\bx(\cdot)$, and its movement cost is at most $8$ times the movement cost of $\bx(\cdot)$. However, it may have fractional mass at inner nodes of the tree. To prevent this, any movement of server mass can be deferred until it reaches a leaf, since placing server mass at internal nodes is never necessary for serving requests.
\end{proof}

We then observe that barely fractional algorithms can handle a small (compared to $m$) additive overhead, as exposed in the following proposition.

\propAdditive*

\begin{proof}[Proof of Proposition \ref{prop:additive}.]
Let $\bz'(\cdot)$ be the variant of $\bz(\cdot)$ that ignores requests that are already satisfied. Concretely, we will have $\bz'(t) = \bz'(t-1)$ if $z'_{\rho(t)}(t-1) \geq 1$. We denote by $\rho'(\cdot)$ the subsequence of $\rho(\cdot)$ where we remove all such superfluous requests. The fractional algorithm $\bz'(\cdot)$ is formally defined by the value of algorithm $\bz(\cdot)$ on the request sequence $\rho'(\cdot)$. Denoting by $T'$ the length of $\rho'(\cdot)$, i.e., the number of requests that are not superfluous, we have
\begin{equation}
    \frac{1}{m}T'\leq \Cost(\bz'(\cdot),\rho(\cdot)) = \Cost(\bz(\cdot),\rho'(\cdot))\leq c \OPT(\rho'(\cdot)) + \frac{1}{2m}T' + c_0,
\end{equation}
where the first inequality comes from the fact that any movement entails a cost of at least $\frac{1}{m}$.  Thus,
\begin{equation}
    \frac{1}{2m}T' \leq c \OPT(\rho'(\cdot))+c_0.
\end{equation}
Therefore, observing that $\OPT(\rho'(\cdot))\leq \OPT(\rho(\cdot))$ because $\rho'(\cdot)$ is a subsequence of $\rho(\cdot)$, we have
\begin{equation}
    \Cost(\bz'(\cdot),\rho(\cdot)) \leq 2c \OPT(\rho(\cdot)) + 2c_0.
\end{equation}
Since superfluous requests can be eliminated in constant time, the computational complexity of $\bz'(\cdot)$ is no larger than the computational complexity of $\bz(\cdot)$.  
\end{proof}

\section{\texorpdfstring{Barely random $k$-server}{Barely random k-server}}\label{sec:barely-rand}

\barelyrandom*

\begin{proof}
We first prove the result for the line metric $\Mcal  = ([n],d(\cdot,\cdot))$, with $\forall u,v\in [n]^2\colon d(u,v) = |u-v|$.
We view this as a tree rooted at $1\in[n]$, and for convenience, other than our usual convention, we consider internal nodes part of the metric space. (Recall that both models are equivalent by adding leaves at distance $0$.)
For a configuration $\bz$, we have that $z_u$ denotes the mass that is located `below' $u\in[n]$, whereas $z_{[u]} = z_u-z_{u+1}$ is the mass located at $u$ (with the convention that $z_{n+1}=0$). Then, given an $m$-barely fractional algorithm $\bz(\cdot)$ for any $i\in \{1,\dots, m\}$,  at any time $t$, we define the deterministic $k$-server configuration $\Rcal_i(t)$ as
\begin{equation}
    \Rcal_i(t) = \left( \max \{u\in [n] \colon \left\lfloor z_u(t)+\frac{i-1}{m}\right\rfloor =h\}\right)_{h\in \{1,\dots, k\}}.
\end{equation}
We observe that $\Rcal_i(t)$ contains always exactly $k$ points of $[n]$, because $z_1(t) = k$, and $z_{n+1}(t)=0$. Also, we note that the randomized algorithm $\Rcal(\cdot)$ effectively serves the incoming requests, i.e., $\forall i \in [m] \colon u\in \Rcal_i(t)$. This is because, when the request is at $u=\rho(t)$, we must have $z_{[u]}(t)=1$ and thus $z_u(t) = z_{u+1}(t)+1$. Finally, we study the movement cost of $\Rcal(\cdot)$. We note that when a unit mass of $\frac{1}{m}$ moves in $\bz(t)$ along one edge, at most one configuration $\Rcal_i(t)$ is affected, moving by a distance of at most one. This shows that the movement cost of $\Rcal(\cdot)$ is less than the movement cost of $\bz(\cdot)$, and it finishes the proof.

We then turn to the proof in the case where the tree is an HST. Adapting the proof of \cite{bansal2015polylogarithmic} is straightforward, and we only describe the main ideas. We will make sure that the $m$-barely random configuration $\Rcal(\cdot) = (\Rcal_1(\cdot),\dots, \Rcal_m(\cdot))$ is \textit{consistent} with $\bz(\cdot)$ at all times, i.e., for all $t>0$ and $u\in V$,
\begin{equation*}
    z_u(t) = \frac{1}{m}\sum_{i\in [m]}\sum_{v\in L_u} \oneb(v\in \Rcal_i(t)).
\end{equation*}
We further say that a configuration $\Ccal$ is \emph{balanced} with respect to $\bz$ if it satisfies
$\forall u \in V\colon n_u(\Ccal) = \sum_{v\in L_u} \oneb(v\in \Ccal) \in \{\floor{z_u},\ceil{z_u}\}$, where $n_u(\Ccal)$ counts the number of servers below $u$ in $\Ccal$.  The result is obtained by induction by applying the following lemma. 
\begin{lemma}\label{lem: bansal}
Consider two $m$-barely fractional configurations $\bz$ and $\bz'$ and some $m$-barely random configuration $\Rcal = (\Rcal_1,\dots, \Rcal_m)$ that is consistent and balanced with respect to $\bz$. Then one can find in polynomial time an $m$-barely random configuration $\Rcal' =(\Rcal_1',\dots, \Rcal_m')$ that is consistent and balanced with respect to $\bz'$ and such that $\frac{1}{m}\sum_{i\in [m]}d(\Rcal_i,\Rcal_i') =\Ocal( {\normalfont\OT}(\bz,\bz'))$. 
\end{lemma}
\begin{proof}[Proof sketch.]
Observe that it suffices to consider the elementary case where: 
\begin{equation*}
    z_{\ell'}' \gets z_{\ell'} +1/m\quad \text{and} \quad
    z_\ell' \gets z_\ell -1/m,
\end{equation*}
for two leaves $\ell,\ell'\in L$. 
We first construct $\bar \Rcal$ that is consistent with $\bz'$, but imbalanced.
If there exists $i\in [m]$ such that $\ell\in \Rcal_i$ and $\ell'\not\in \Rcal_i$, we simply transfer the server at $\ell$ in $\Rcal_i$ to $\ell'$. Otherwise, since $z_\ell>0$ and $z_{\ell'}<1$, we can pick $i\neq j\in [m]$ such that $\ell,\ell' \in\Rcal_i$ and $\ell,\ell' \not\in \Rcal_j$. We denote by $u$ the lowest common ancestor of $\ell$ and $\ell'$. Because both configurations are balanced with respect to $\bz$, their numbers of servers below $u$ differ by at most $1$, and we can pick a leaf $\ell''$ that is a descendant of $u$ such that $\ell''\in \Rcal_j$ and $\ell''\not\in\Rcal_i$. We then proceed to the following updates,  
\begin{equation*}
    \bar \Rcal_i \gets \Rcal_i \cup \{\ell''\} \setminus \{\ell\} \quad \text{and}\quad  
    \bar \Rcal_j \gets \Rcal_j \cup \{\ell'\} \setminus \{\ell''\}.
\end{equation*}
The total movement cost for this operation is at most $2\OT(\bz,\bz')$, and the configuration $\bar \Rcal$ obtained is consistent with $\bz'$. However, the configuration $\bar \Rcal$ might not be balanced with respect to $\bz'$, specifically, descendants of $u$ might have become imbalanced after this transformation, where $u$ is the lowest common ancestor of $\ell$ and $\ell'$. We therefore need to transform $\bar \Rcal$ into a configuration $\Rcal'$ that is balanced, while incurring a limited movement cost. For this, the idea is to balance all descendants $v$ of $u$, proceeding inductively from top to bottom. Concretely, this is achieved by exchanging points between pairs of configurations $(\Rcal_{i},\Rcal_{j})$ that satisfy $n_v(\Rcal_i)>\ceil{z_u}$ and $n_v(\Rcal_j)<\floor{z_u}$, thus strictly reducing the balance gap
$$G(\Rcal,\bz) = \frac{1}{m}\sum_{v\in V}w_v \sum_{i\in [m]} \min\{|n_v(\Rcal_i)-\floor{z_v}|,|n_v(\Rcal_i)-\ceil{z_v}|\},$$
while preserving the consistency property. We refer to \cite{bansal2015polylogarithmic} for complete details. This procedure runs in polynomial time because, when a node $v$ is being balanced, the value of $\sum_{i\in [m]}\min\{|n_v(\Rcal_i)-\floor{z_v}|,|n_v(\Rcal_i)-\ceil{z_v}|\}$ decreases by at least $2$ at each balancing step.
\end{proof}

\end{proof}

\section{Polynomial time algorithm with additive cost}\label{sec:ccomplexity}
In this section, we explain how we adapt the fractional $k$-server algorithm of \cite{BuchbinderGMN19} to run in polynomial time, albeit at the expense of an additive cost, as stated in Proposition \ref{prop:approximate-smile}. 

\approximateSmile*

The algorithm in \cite{BuchbinderGMN19} operates by finding a solution to a constrained convex optimization problem, which is restated in Equation \eqref{eq:BregmanProjection} below. As the solution can contain irrational numbers, we cannot solve it exactly in finite time or space, so we aim for an approximate solution. The constraint set, which is known in the literature as the assignment polytope \cite{BubeckCLLM18} or anti-server polytope \cite{BuchbinderGMN19}, is defined by an exponential number of constraints, but admits a fast separation oracle (see Remark \ref{rem:sep-oracle-polytope} in Appendix~\ref{ap:ellipsoid}). Thus, we can use the ellipsoid algorithm to find an approximate solution of this problem in polynomial time. 

We first recall some definitions from~\cite{BuchbinderGMN19} relevant to our discussion.
Some of the notation in~\cite{BuchbinderGMN19} is for the more general case where an online algorithm with $k$ servers competes against an offline algorithm with $h$ servers, for $h\le k$. For our purposes, the special case $h=k$ suffices.

For a node $u\in V$, let $n_u=|L_u|$ be the number of leaves in the subtree rooted at $u$. Let $\chi:=\{(u,j)\mid u\in V, j\in [n_u]\}$, and for $u\in V\setminus L$ let $\chi_u:=\{(v,j)\in\chi\mid v\in C_u\}=\{(v,j)\mid v\in C_u, j\in[n_v]\}$. The anti-server polytope employed in~\cite{BuchbinderGMN19}, originally introduced in~\cite{BubeckCLLM18}, is defined as
\begin{align*}
    P:= \bigg\{\bx\in[0,1]^\chi\,\,\bigg\vert \qquad x_{rj}&\ge \1_{j>k}&&\forall j\in [n_r],\\
    \sum_{i\le |S|}x_{ui}&\le \sum_{(v,j)\in S}x_{vj}&&\forall u\in V\setminus L,\, S\subseteq \chi_u\bigg\}.
\end{align*}
For intuition, an integral $k$-server configuration can be encoded by an ``anti-server configuration'' $\bx\in P$ by setting $x_{uj}=0$ if the subtree rooted at $u$ contains at least $j$ servers and $x_{uj}=1$ otherwise. In particular, a leaf $\ell\in L$ contains a server if $x_{\ell1}=0$.

Let $\delta=\frac{1}{2k+1}$ and $P_\delta:=\{\bx\in P\mid x_{\ell1}\ge \delta\text{ for all }\ell\in L\}$. The online algorithm in~\cite{BuchbinderGMN19} maintains an anti-server configuration $\bx(t)$ in $P_\delta$ that satisfies $\sum_{\ell\in L}x_{\ell1}(t)=n-k$. This is then converted into a leaf measure $\bz(t)\in\Zcal$ of mass $k+\frac{1}{2}$ by setting $z_\ell(t):=\frac{1-x_{\ell1}(t)}{1-\delta}$ for $\ell\in L$ and $z_u=\sum_{\ell\in L_u}z_\ell$, thus yielding a fractional $\left(k+\frac{1}{2}\right)$-server algorithm. A simple transformation (first shown in~\cite{BubeckCLLM18}, analogous to step \ref{it:kPlusHalfToK} in our proof of Proposition~\ref{prop:frac-barely-frac}) then shows that any fractional $\left(k+\frac{1}{2}\right)$-server algorithm can be converted into a fractional $k$-server algorithm while increasing the cost by only a constant factor.

When a request arrives at a leaf $\rho(t)\in L$ and the previous anti-server configuration is $\bx(t-1)\in P$, the algorithm of~\cite{BuchbinderGMN19} chooses the new anti-server configuration as
\begin{align}
    \bx(t):= \argmin_{\substack{\bx\in P_\delta\\ x_{\rho(t)1}=\delta}} D(\bx\parallel \bx(t-1)),\label{eq:BregmanProjection}
\end{align}
where
\begin{align}\label{eq:divergence}
    D(\bx\parallel \bx'):=\sum_{u\in V\setminus\{r\}}w_u\sum_{j}\left((x_{uj}+\delta)\log\frac{x_{uj}+\delta}{x'_{uj}+\delta} - x_{uj} + x'_{uj}\right).
\end{align}
Observe that the associated leaf measure $\bz(t)$ serves the request, since $z_{\rho(t)}=\frac{1-x_{\rho(t)1}}{1-\delta}=1$.

In the following Lemma,
\begin{align*}
    \left\|\bx(t)-\bx(t-1)\right\|_w^+:=\sum_{u,i}w_u\left(x_{ui}(t)-x_{ui}(t-1)\right)^+
\end{align*}
denotes the `positive movement cost' associated with the transition from $\bx(t-1)$ to $\bx(t)$. Over any sequence of steps, the total positive and negative movement are equal up to a bounded additive error, and hence the sum of these terms over all time steps upper bounds the movement cost of $\bz(\cdot)$ up to a constant factor.

\begin{lemma}\label{lem:potential equation}
    Let the underlying tree be a $\tau$-HST with $\tau\ge 10$. Let $\by(t)\in P$ be the anti-server configuration at time $t$ corresponding to an optimal offline algorithm. For any configuration~$\bx(t-1)\in P_\delta$ at time $t-1$, the new configuration $\bx(t)$ defined in~\eqref{eq:BregmanProjection} satisfies
    \begin{align*}
        \left\|\bx(t)-\bx(t-1)\right\|_w^+ &+ Q(\bx(t),\by(t)) - Q(\bx(t-1),\by(t-1))\le \Ocal\left(\log^2k\right)\cdot\Delta_t\OPT,
    \end{align*}
    where $Q(\bx,\by)$ is a potential function that is $O(D\log^2 k)$-Lipschitz in any $x_{ui}$, and $\Delta_t\OPT$ is the cost of the offline algorithm at time $t$.
\end{lemma}
\begin{proof}
    The function $Q$ is a combination of several potential functions used in \cite{BuchbinderGMN19}. We prove our lemma by combining several inequalities that have been derived in \cite{BuchbinderGMN19}. Since the proof of some of these inequalities is rather involved, we do not repeat it here and instead refer to the original paper. When citing specific results from \cite{BuchbinderGMN19}, we will refer to the numbering used in the corresponding full version of the paper \cite{BuchbinderGMN18arxiv}.
    
    First, define the potential function
    \begin{align*}
        \Phi(\by\parallel\bx)&:= \sum_{u\in V\setminus\{r\}}w_u\sum_{j}(y_{uj}+\delta)\log\frac{y_{uj}+\delta}{x_{uj}+\delta}.
    \end{align*}
    where $\by\in P$ and $\bx\in P_\delta$ refer to the anti-server configurations of the offline and online algorithm, respectively. By \cite[Lemma 3.6]{BuchbinderGMN18arxiv}, when the offline algorithm moves from $\by(t-1)$ to $\by(t)$, the change in $\Phi$ can be bounded by
    \begin{align}
        \Phi(\by(t)\parallel\bx(t-1)) - \Phi(\by(t-1)\parallel\bx(t-1)) \le (1+\delta)\log\left(1+\frac{1}{\delta}\right)\cdot\Delta_t\OPT.\label{eq:BGMNOpt}
    \end{align}
    
    Further, \cite[Theorem~4.1]{BuchbinderGMN18arxiv} shows for some constants $c_1, c_2, c_3> 0$ that
    \begin{align}
        \left\|\bx(t)-\bx(t-1)\right\|_w^+ \le c_1\log\frac{k}{\delta}\cdot \sum_j A_{rj}^t+\Psi(\bz(t))-\Psi(\bz(t-1)),\label{eq:BGMNThm4.1}
    \end{align}
    where $\bz(t)$ is the leaf measure of mass $k+\frac{1}{2}$ corresponding to $\bx(t)$ and
    \begin{align*}
        \Psi(\bz):=c_2\sum_{u}\frac{w_u}{1-\delta}\int_{z_u(0)}^{z_u}\ln\left(1+\frac{z}{\delta}\right)dz - c_3\sum_{u}w_u\left(z_u-\frac{2}{3}z_{p(u)}\right)^+
    \end{align*}
    and $A_{rj}^t$ are quantities that by~\cite[Inequality~(3.13)]{BuchbinderGMN18arxiv} satisfy
    \begin{align}
        \Phi(\by(t)\parallel\bx(t-1)) - \Phi(\by(t)\parallel\bx(t))\ge \frac{1}{3}\sum_{j}A_{rj}^t + \frac{2}{3}(W(\bx(t)) - W(\bx(t-1))),\label{eq:BGMNA}
    \end{align}
    for
    \begin{align*}
        W(\bx)&:= \sum_{u}\sum_{j}w_u x_{uj}.
    \end{align*}
    Setting
    \begin{align*}
        Q(\bx, \by) := 3c_1\log\frac{k}{\delta}\cdot\Phi(\by\parallel\bx) + 2c_1\log\frac{k}{\delta}\cdot W(\bx) - \Psi(\bz)
    \end{align*}
    and combining inequalities~\eqref{eq:BGMNOpt}, \eqref{eq:BGMNThm4.1} and \eqref{eq:BGMNA} yields
    \begin{align*}
        \left\|\bx(t)-\bx(t-1)\right\|_w^+ &+ Q(\bx(t),\by(t)) - Q(\bx(t-1),\by(t-1)) \\
        &\le 3c_1\log\frac{k}{\delta}\cdot(1+\delta)\log\left(1+\frac{1}{\delta}\right)\cdot\Delta_t\OPT \\
        &= \Ocal\left(\log^2k\right)\cdot\Delta_t\OPT.
    \end{align*}
    Using $w_u\le D$ and $1/\delta=O(k)$, it is easy to check that $Q$ is $O(D\log^2k)$-Lipschitz in any $x_{ui}$.
\end{proof}

Equipped with Lemma \ref{lem:potential equation}, we can now conclude this section with the proof of Proposition \ref{prop:approximate-smile}.
\begin{proof}[Proof of Proposition \ref{prop:approximate-smile}] The algorithm is defined by induction by iteratively obtaining an approximate solution of Equation \eqref{eq:BregmanProjection} with the Ellipsoid method (Corollary \ref{cor:ellipsoid2} in Appendix \ref{ap:ellipsoid}) with a precision that we will specify later. This defines a sequence of points in the anti-server polytope $\bx^\epsilon(t)$, which are then transformed into a fractional $k$-server configuration $\bz(t)$ via the transformations explained above, with a loss of only a constant factor in the movement cost in the process.  

By Corollary \ref{cor:ellipsoid2}, the sequence $\bx^{\epsilon}(t)$ satisfies the following approximate version of Lemma \ref{lem:potential equation}: 
\begin{align*}
        \left\|\bx^\epsilon(t)-\bx^\epsilon(t-1)\right\|_w^+ + Q(\bx^\epsilon(t),\by(t)) - Q&(\bx^\epsilon(t-1),\by(t-1))\\
        &\le \Ocal\left(\log^2k\right)\cdot\Delta_t\OPT + \Ocal(\epsilon' D\log^2 k),
    \end{align*}
where we used the fact that $Q(\bx,\by)$  is $O(D\log^2 k)$-Lipschitz in its first argument, and that $||\cdot||_w^+$ is $O(D)$-Lipschitz. By telescopic sum, 
\begin{align*}
    \sum_{t\leq T}\left\|\bx^\epsilon(t)-\bx^\epsilon(t-1)\right\|_w^+ \leq \Ocal(\log^2 k)\OPT(\rho(\cdot)) + \Ocal(\epsilon'D\log^2 k T) + \Ocal(D\log^2 k),
\end{align*}
where we use that $Q(\bx,\by)$ is bounded by $\Ocal(D\log^2 k)$. For the right choice of precision in our application of Corollary \ref{cor:ellipsoid2} of order $\epsilon/(D\log^2 k)$, we thus have that 
\begin{align*}
    \Cost(\bz(\cdot),\rho(\cdot))\leq \Ocal(\log^2 k)\OPT(\rho(\cdot))+\epsilon T + \Ocal(D\log^2 k).&\qedhere
\end{align*}
\end{proof}

\paragraph{Acknowledgments.} We thank the anonymous reviewers for their constructive feedback. CC is funded by the European Union (ERC, CCOO, 101165139). Views and opinions expressed are however those of the author(s) only and do not necessarily reflect those of the European Union or the European Research Council. Neither the European Union nor the granting authority can be held responsible for them.

\bibliographystyle{alpha}
\bibliography{sample}

\newpage
\appendix

\section{Ellipsoid method}\label{ap:ellipsoid}
In this section, we briefly explain how the ellipsoid method allows us to provide approximate solutions for problems of the form of \eqref{eq:BregmanProjection} in polynomial time. We start with some definitions from \cite{grotschel2012geometric}.

\begin{definition}[Approximate oracles \cite{grotschel2012geometric}]
    An approximate separation oracle for a closed convex set $K\subset \Rbb^d$ is an algorithm that takes as input a rational number $\gamma>0$ and a rational point $\by\in\Qbb^d$, and that, in polynomial time in the encoding length of the inputs either reports that $\by$ is at distance at most $\gamma$ from $K$, or returns a vector $\bc\in \Qbb^d$ satisfying $||\bc||_\infty = 1$ and $\forall \bx\in K \colon \bc^T \bx \leq \bc^T\by +\gamma$. An approximate first order oracle for a convex differentiable function $f:\Rbb^d\rightarrow \Rbb$ is an algorithm that takes as input a rational number $\gamma>0$ and a rational point $\by\in\Qbb^d$, and that, in polynomial time in the encoding lengths of the inputs returns a vector $\bg\in \Rbb^d$ that satisfies $||\bg - \nabla f(\by)||\leq \gamma$. An approximate zeroth-order oracle for $f(\cdot)$ is an algorithm that, for the same input, returns in polynomial time in the encoding lengths of the inputs a value $v$ such that $|f(\by)-v|\leq \gamma$. 
\end{definition}

\begin{theorem}[Central cut ellipsoid, Theorem 3.2.1 of \cite{grotschel2012geometric}]\label{th:central-cut}
    There is an algorithm that takes as input a convex set $K\subset \Rbb^d$ contained in a ball of radius $R$ represented by an approximate separation oracle, and a precision $\epsilon>0$, and in polynomial time in $\left(d,~\log R,~ \log 1/\epsilon\right)$ either outputs a point that is at distance at most $\epsilon$ from $K$ or an ellipsoid of volume at most $\epsilon^{2d}$ that contains $K$. 
\end{theorem}
\begin{remark}
    In the original statement of Theorem \ref{th:central-cut}, the volume of the ellipsoid for the failure case is $\epsilon$, but it can be strengthened to $\epsilon^{2d}$ for a computational cost of same order by running the algorithm with $\epsilon' =\epsilon^{2d}$. In practice, if the polytope is full-dimensional, its volume is lower-bounded by an exponential function of its encoding length (see \cite[Lemma 3.1.35]{grotschel2012geometric}), and the failure case is not raised. 
\end{remark}
We will utilize the following corollary, which concerns the approximate optimization of convex functions. 
\begin{corollary}\label{cor:ellipsoid} There is an algorithm that takes as input a precision $\epsilon>0$, a closed convex set $K\subset \Rbb^d$ contained in a ball of radius $R$ represented by an approximate separation oracle and satisfying that the volume of the intersection of $K$ with any ball centered on a point of $K$ of radius $r$ is at least $\Omega(r^d)$, and a convex function $f(\cdot)$ of $\Rbb^d$ represented by its approximate zeroth-order and first-order oracles, that is $L$-Lipschitz, $\alpha$-strongly convex, and a closed interval of width $w$ to which the minimum value of $f(\cdot)$ belongs, and that, in time polynomial in $(d, \log R, \log 1/\epsilon, \log 1/\alpha, \log L)$ outputs a point $\bx^{*,\epsilon}$ satisfying $||\bx^{*,\epsilon}-\bx^*||\leq \epsilon$ where $\bx^* = \argmin_{\bx\in K}f(\bx)$.
\end{corollary}
\begin{proof}
Define $f^* = \min_{\bx\in K}f(\bx)$. for a fixed threshold $A\in \Rbb$, consider the closed convex set defined by 
$$K_A = \{\bx \colon f(\bx)\leq A\}\cap K.$$
We start by assuming that we have $A\in [f^*+\epsilon',f^*+2\epsilon']$, where $\epsilon' = \alpha\epsilon^2/8$, and that we are given $\bx^{*,\epsilon}$ that is at a distance at most $\epsilon/2$ from a point $\bx_A$ in $K_A$. Then, by strong convexity of $f(\cdot)$, we have $\alpha||\bx^*-\bx_A||^2\leq f(\bx_A)-f(\bx^*) \leq \alpha\epsilon^2/4$ so that $||\bx_A-\bx^*||\leq \epsilon/2$. And by triangle inequality, $||\bx^{*,\epsilon}-\bx^*|| \leq \epsilon/2 +\epsilon/2=\epsilon$.

We now explain how to find a value of $A\in [f^*+\epsilon',f^*+2\epsilon']$, and to obtain the corresponding point $\bx^{*,\epsilon}$ in polynomial time. We will proceed by dichotomy. We know that the minimum value of the function is within an interval of width $w$. We maintain the interval $[A^-_i,A^+_i]$ of candidates for $A$, starting with $[A_0^-,A_0^+]$ being that interval of width $w$. We run the central cut ellipsoid algorithm for $K_{A_i}$ where $A_i$ is the midpoint of the interval, $A_i = (A_i^++A_i^-)/2$, and, if it returns a solution, we update the interval to $[A_i^-,A_i]$ and otherwise to $[A_i,A_i^+]$. We note that in only $\Ocal(\log w + \log 1/\epsilon')$ 
steps, the width of the interval is smaller than $\epsilon'/2$. We also note that $A_i^-$ cannot be greater than $f^*+\epsilon'$, because, for any $A\geq f^*+\epsilon'$, we have that the volume of $K_A$ is in $\Omega((\epsilon'/L)^d)$ because it contains the intersection of $K$ with a ball of radius $\epsilon'/L$. Finally, running the central cut ellipsoid algorithm on $A_i^+\in [f^*+\epsilon',f^*+2\epsilon']$ returns the desired output $\bx^{*,\epsilon}$. 

To conclude, we need to argue that for any value of $A$, we can run the central cut ellipsoid algorithm for $K_A$. Note that $K_A$ is contained in a closed convex set included in a ball of radius $R$ and that an approximate separation oracle for $K_A$ is given at $\by \in \Qbb^d$ by (1) an approximate separation oracle for $K$ at $\by$ if $\by \not\in K$ (2) an approximate gradient of $f(\cdot)$ at $\by$ if $f(\by)\geq A$. Thus, applying Theorem~\ref{th:central-cut}, the central cut ellipsoid method applied to $K_A$ and $\epsilon/2$ outputs either a point that is at distance at most $\epsilon/2$ from $K_A$ or a certificate that the volume of $K_A$ is less than $(\epsilon/2)^{2d}$.
\end{proof}

We now apply this theorem to real functions of the form $f(\bx) = D(\bx\parallel\bx(t-1))$ for the divergence defined in \eqref{eq:divergence} and some fixed vector $\bx(t-1)$ and to the anti-server polytope $P_\delta$ (where we drop the variable $x_{\rho(t)1}$ which is set equal to $\delta$). The dimension $d$ of the polytope is bounded by $n^2$, the function $f(\cdot)$ is $L$-Lipschitz with $L = O(D\log k)$ and $\alpha$-strongly convex with $\alpha = \Omega(1)$. There exists a separation oracle for the anti-server polytope (Remark \ref{rem:sep-oracle-polytope}), and, furthermore, by Taylor expansions of the $\log(\cdot)$ on the interval $[\delta,1+\delta]$, approximate zeroth and first-order oracle for $f(\cdot)$ are known. Finally, we also note that for any point $\bx$ at a distance $\eta$ from the antiserver polytope, there is a polynomial time procedure (specifically, enforcing the constraints from top to bottom in the tree) that returns a point $\bx'$ in the antiserver polytope at distance at most $O(\eta)$ from $\bx$. Therefore, applying Corollary \ref{cor:ellipsoid}, we have the following result.
\begin{corollary}\label{cor:ellipsoid2} For any $\epsilon>0$, for any HST with $n$ nodes and diameter $D$, and for any element $\bx(t-1)$ of the antiserver polytope, there is a polynomial time algorithm that provides an approximate solution $\bx^{\epsilon}$ of \eqref{eq:BregmanProjection} satisfying $||\bx(t)-\bx^\epsilon||\leq \epsilon$ in time polynomial in $(\log 1/\epsilon,n,\log D)$.
\end{corollary}

\begin{remark}[Separation oracle for the antiserver polytope]\label{rem:sep-oracle-polytope} Let $n_u$ be the number of leaves below $u$. Recall that $\chi(u)$ is the set of pairs $(v,i)$ where $v$ is a child of $u$ and $i \leq n_v$. Let $(v_1,i_1), (v_2, i_2), \dots$ be the list of elements of $\chi(u)$ in increasing order of their $x$-values. Note that a constraint between u and its children is violated if and only if the constraint of a set S that is a prefix of this list is violated. Since there are only $n_u$ prefixes, for each $u$, we only need to check polynomially many constraints to find a violated constraint, if one exists.
\end{remark}

\end{document}